\documentclass[a4paper,final]{amsart}

\usepackage[T1]{fontenc}
\usepackage[utf8]{inputenc}
\usepackage{amsfonts}
\usepackage{amssymb}
\usepackage{amsthm}
\usepackage{mathtools}
\usepackage[scaled]{helvet}
\usepackage{mathpazo}
\usepackage{enumitem}
\usepackage[scr=rsfso,scrscaled=1.1]{mathalpha}

\newcommand{\MM}{{\mathbb M}}
\newcommand{\RR}{{\mathbb R}}
\newcommand{\VV}{{\mathbb V}}

\newcommand{\cC}{{\mathcal C}} 
\newcommand{\cR}{{\mathcal R}} 
\newcommand{\cS}{{\mathcal S}} 
\newcommand{\del}{\partial}
\newcommand{\eps}{\varepsilon}
\newcommand*{\dd}{\mathrm{d}}
\newcommand*{\Lie}{\mathscr{L}}

\usepackage[style=phys,sorting=nyt,
                articletitle=true,
                biblabel=brackets,
                chaptertitle=true,
                pageranges=false,
                hyperref=true,
                maxcitenames=3,
                url=true,
                autocite=superscript] {biblatex}

\usepackage{hyperref}
\hypersetup{linkcolor=black,
            urlcolor=blue,  
            breaklinks=true,
            colorlinks=true,
            citebordercolor=0 0 0,
            filebordercolor=0 0 0,
            linkbordercolor=0 0 0,
            menubordercolor=0 0 0,
            urlbordercolor=0 0 0
} 
\usepackage[color]{showkeys}

\newtheorem{theorem}{Theorem}[section]
\newtheorem{prop}[theorem]{Proposition}

\newtheorem{corr}[theorem]{Corollary}

\numberwithin{equation}{section}

\newcommand*{\hnabla}{\widehat{\nabla}}

\newcommand*{\tg}{\tilde{g}}

\begin{document}
\title{A note on the conformal geometry of the 2-sphere}

\author[J. Frauendiener]{Jörg Frauendiener}
\email{joerg.frauendiener@otago.ac.nz}
\address{Department of Mathematics and Statistics, University of Otago, New Zealand}

\begin{abstract}
  We discuss the geometry of the two-dimensional sphere. Our perspective focuses on the set of conformal factors that rescale a given metric to a round metric. We find that there exists a unique symmetric and trace-free tensor on the sphere which complements the Ricci tensor to a divergence-free symmetric tensor. Using this tensor, we derive a linear second-order PDE whose solution space carries a natural Minkowski structure. We show that there are in fact infinitely many such spaces, all related by isometries which can be thought of as being parametrised by ``4-vectors''. Finally, we show that every 2-sphere can be isometrically embedded in a Minkowski space in a canonical way.
\end{abstract}
\maketitle

\section{Introduction}
\label{sec:introduction}

It is well known that the group of Möbius transformations is isomorphic to the identity component of the Lorentz group. But why is that the case? The Möbius group appears as automorphisms of the complex structure of the Riemann sphere, while the Lorentz group is the invariance group of Minkowski space, the arena of special relativity. Where is the Minkowski space associated with the Möbius transformations?

To be sure, we can set up a relationship between the sphere and Minkowski space by considering the set of null directions attached to the origin in Minkowski space. This defines the celestial sphere (see Penrose and Rindler~\cite{Penrose:1984a}). The induced action of the Lorentz group on the celestial sphere coincides with the Möbius group. However, this construction seems like an afterthought. One should be able to define a Minkowski space in a canonical way from the geometric properties of the 2-sphere alone. But how can this be achieved?

This question was one of the motivations for this note. Other reasons arose from the study of the geometry of null-infinity, where certain conformally invariant constructions became apparent in relation to foliations by spherical ``cuts'' (see~\cite{Frauendiener:2022} and a forthcoming paper). It soon became apparent that these properties were not inherent to the null-structure of null-infinity but were due to the conformal structure of the spherical cuts.

In this note, we will discuss the conformal geometry of the 2-sphere from an alternative perspective. The focus will be on the relationship between the generic sphere and a round sphere, by which we will always mean the 2-sphere equipped with a metric of constant Gauß curvature. This relationship is mediated by functions on the sphere which rescale the given metric to a round metric. This round metric, in a sense, provides a rather restricted background against which the generic metric can be viewed. Each conformal factor provides its own lens for viewing this background, thereby offering a different perspective.

The interplay between these conformal factors yields a uniquely defined symmetric, trace-free tensor that complements the Ricci tensor (which is pure trace) of the given metric to form a divergence-free symmetric tensor. This unique tensor provides a different measure for the deviation of the metric from roundness, which is different from its curvature. It has appeared in the literature before, at least as early as 1977 in the seminal work by Geroch~\cite{Geroch:1977} on the asymptotic structure of space-times, in which he introduced the ``$\rho$-tensor'' as a means to compensate for the complicated transformation behaviour of the Ricci tensor of the metric induced on a cut. The Geroch tensor is fundamental for discussing the structure of null-infinity in an invariant way.

This tensor is indeed an intrinsic property of the sphere, and therefore equally fundamental for capturing some of its essential geometric properties. We will show below that all conformal factors to round metrics satisfy a linear PDE on the sphere, which is constructed from the distortion tensor. The solution space of this equation turns out to be a natural Minkowski space, a 4-dimensional real vector space with a naturally defined inner product with Lorentzian signature.

Nothing in this note is really new. However, we hope the perspective on the geometry of the 2-sphere that we exhibit here is fruitful for further developments in GR. As one example, we propose an application to the discussion of quasi-local quantities.

The paper is structured as follows. We provide the fundamental equations and relationships that we use in Sect.~\ref{sec:basic-equations}. In Sect.~\ref {sec:resc-round-metr} we derive the properties of the conformal factors that we mentioned above, and in Sect.~\ref{sec:solving-phi-equation} we define the intrinsic Minkowski space. Next, in Sect.~\ref{sec:invariance-group-1}, we discuss transformations which leave the construction invariant and, finally, in Sect.~\ref{sec:canonical-embeddings} we construct a canonical embedding of every 2-sphere into an appropriate Minkowski space. Sect.~\ref{sec:conclusion} provides summary and outlook.

A quick note on our notation and conventions. We are concerned here with the 2-sphere. The round unit-sphere in $\RR^3$ is defined by
\[
  S^2:=\{ (X^A)_{A=1:3} = (x,y,z)\in \RR^3 : x^2 + y^2 + z^2 = 1\}.
\]
We assume $S^2$ is equipped with the usual metric inherited from $\RR^3$ \emph{except} that it is \emph{negative} definite
\[
  q = -\delta_{AB}\dd X^A\dd X^B.
\]
The connection compatible with the round metric $q_{ab}$ is denoted by $\del_a$, so that $\del_cq_{ab}=0$. The sphere has constant Gauß curvature $K=1$. We refer to any metric on $S^2$ with constant Gauß curvature as a round metric. A metric with $K=1$ is also referred as a unit-metric. In general, for an arbitrary metric $g_{ab}$ on $S^2$ with covariant derivative $\nabla_a$, its Riemann tensor $r_{abc}{}^d$ can be written in the form
\[
  r_{abc}{}^d = K\eps_{ab}\eps_c{}^d = K\left(g_{ac}\delta_b{}^d - \delta_a{}^dg_{bc}\right) \implies r_{ab} = K g_{ab}, \quad r_a{}^a = 2K.
\]
Here, we have also introduced the area-form $\eps_{ab}$. Otherwise, we use the conventions of~\cite{Penrose:1984a} throughout.

\section{The basic equations}
\label{sec:basic-equations}

In this section, we collect the main relationships that we will need in the discussions below. Let $(\cS,g_{ab})$ be a generic spherical surface with metric $g_{ab}$ and compatible connection $\nabla_a$. Furthermore, let $v_a$ be any 1-form on $S^2$. Then we find 
\[
  \begin{aligned}
    \nabla_a\nabla^ev_e &= \nabla_e\nabla_av^e + r_{aec}{}^ev^c = \nabla^e\nabla_av_e + Kv_a \\
    &= \nabla^e(\nabla_{[a}v_{e]}) + \nabla^e(\nabla_{\{a}v_{e\}}) + \frac12 \nabla_a\nabla^ev_e + Kv_a.
\end{aligned}
\]
Thus, with $v_a=\nabla_af$ and $\Delta=\nabla^e\nabla_e$, we obtain the important identity
\begin{equation}
  \label{eq:1}
  \nabla_a(\Delta f - 2 K f)= 2 \nabla^e(\nabla_{\{a}\nabla_{e\}}f) - 2\nabla_aK f.
\end{equation}
Recall the formula for the Gauß curvatures of two conformally related metrics $g_{ab} = \Theta^2\tg_{ab}$\footnote{Note that this is the formula for a negative definite metric. In the positive definite case, there would be a plus sign.}
\begin{equation}
  \tilde{K} = \Theta^{2}\left(K - \nabla^a\Theta_a\right) , \quad \Theta_a := \nabla_a\log\Theta.\label{eq:2}
\end{equation}
Taking the derivative of the equation, we compute
\[
  \begin{aligned}
    \Theta^{-2}\nabla_a\tilde{K} &= 2\Theta_a K + \nabla_a K - 2\Theta_a \nabla^e\Theta_e - \nabla_a\nabla^e\Theta_e \\
    &= 2\Theta_aK + \nabla_a K - 2\Theta_a \nabla^e\Theta_e - \nabla^e(\nabla_{\{a}\Theta_{e\}}) - \frac12 \nabla_a\nabla^e\Theta_e - K \Theta_a \\
    &= K \Theta_a + \nabla_a K - \Theta_a \nabla^e\Theta_e - \frac12 \nabla_a\nabla^e\Theta_e - \nabla^e(\nabla_{\{a}\Theta_{e\}})  - \Theta_a \nabla^e\Theta_e \\
    &= \frac12 \Theta^{-2}\nabla_a\tilde{K} + \frac12 \nabla_a K - \nabla^e(\nabla_{\{a}\Theta_{e\}})  - \nabla^e(\Theta_{\{a} \Theta_{e\}}).
  \end{aligned}
\]
Thus, we found the important relationship between the gradients of the Gauß curvatures
\begin{equation}
  \label{eq:3}
  \Theta^{-2}\nabla_a\tilde{K} - \nabla_a K = -2\nabla^e\left(\nabla_{\{a}\Theta_{e\}} + \Theta_{\{a} \Theta_{e\}}\right) = -2 \nabla^e\left(\Theta^{-1}\nabla_{\{a}\nabla_{e\}}\Theta\right).
\end{equation}
For later use, let us define the symmetric trace-free tensor defined in terms of $\Theta$ by
\[
  \Theta_{ab}:=\Theta^{-1}\nabla_{\{a}\nabla_{e\}}\Theta.
\]
Replacing $\tilde{K}$ on the left-hand side from~\eqref{eq:2}, we obtain
\begin{equation}
  \label{eq:4}
  \nabla_a\left[K\Theta^2 - \Theta \nabla^e\nabla_e \Theta + \nabla^e\Theta \nabla_e \Theta\right] = \Theta^2 \left(\nabla_aK - 2 \nabla^e\Theta_{ae}\right).
\end{equation}

\section{Rescaling to a round metric}
\label{sec:resc-round-metr}

Suppose the metric $\tg_{ab}$ is a round metric with constant Gauß curvature $\tilde{K}=k$. Then the formula~\eqref{eq:2} for the Gauß curvatures yields an equation for the conformal factor $\Theta$ between the two metrics
\begin{equation}
  \label{eq:5}
   k = K\Theta^{2} - \Theta \nabla^a\nabla_a \Theta + \nabla^a\Theta \nabla_a \Theta.
\end{equation}
Furthermore, eq.~\eqref{eq:4} now implies 
\begin{prop}\label{prop:1}
  Suppose $\Theta:\cS \to \RR$ is non-zero everywhere, then the following statements are equivalent:
  \begin{enumerate}[label=(\roman*)]
  \item $\Theta^{-2}g_{ab}$ is a metric with constant Gauß curvature,  
  \item $\Theta$ satisfies the equation
  \begin{equation}
    \label{eq:6}
    \nabla_aK = 2 \nabla^e\Theta_{ae}.
  \end{equation}
\item The expression
  \begin{equation}
    \Sigma(\Theta,\Theta):= K\Theta^{2} - \Theta \nabla^a\nabla_a \Theta + \nabla^a\Theta \nabla_a \Theta\label{eq:7}
  \end{equation}
  is constant on $\cS$.
\end{enumerate}
\end{prop}
This proposition immediately implies the
\begin{theorem}\label{theorem:1}
  Let $(\cS,g_{ab})$ be a 2-dimensional manifold with spherical topology. Then there exists a unique symmetric and trace-free tensor $\cR_{ab}$ such that
  \[
    \nabla^e(r_{ae} - 2 \cR_{ae}) = 0.
  \]
\end{theorem}
\begin{proof}
  Existence follows from Proposition~\ref{prop:1} since $\nabla^er_{ae}=\nabla_aK=2\nabla^e\Theta_{ab}$ for a suitable $\Theta$. To show uniqueness, we assume that $\hat\cR_{ab}$ is another such tensor. Then the difference $d_{ab}=\cR_{ab}-\hat\cR_{ab}$ satisfies the equation $\nabla^ed_{ae}=0$. It is a classical result that there are no symmetric, trace-free tensors on $S^2$ which are also divergence-free, so that this equation only admits the solution $d_{ab}=0$.
\end{proof}
\begin{corr}\label{corr:1}
  Any conformal factor $\Theta$ according to Proposition~\ref{prop:1} satisfies the linear equation
  \begin{equation}
    \label{eq:8}
    \nabla_{\{a}\nabla_{b\}} \Theta = \cR_{ab} \Theta.
  \end{equation}
\end{corr}
To our knowledge, the field $\cR_{ab}$ has made its first appearance in the analysis of the asymptotic structure of space-times by Geroch~\cite{Geroch:1977}. It was used heavily in~\cite{Frauendiener:2022}, where gauge-independent expressions for asymptotic quantities were derived using a conformal variant of the familiar GHP formalism~\cite{Geroch:1973b,Penrose:1984a}. For obvious reasons, it is often referred to as the Geroch tensor or the co-curvature. Since it provides an alternative characterization of the deviation of a spherical metric from roundness, i.e., of its inhomogeneity, one could also call it the \emph{distortion tensor}.

A further obvious property of the distortion tensor is that it complements the Ricci tensor, which is pure trace here, to a symmetric and divergence-free tensor on $\cS$. In some sense, one could regard this tensor as a 2-dimensional Einstein tensor, with $\nabla^a(r_{ab}- 2\cR_{ab})=0$ taking the place of a Bianchi identity.

Let us now consider the behaviour of $\cR_{ab}$ under conformal rescalings of $g_{ab}$. Suppose, $\Theta^{-2}g_{ab}=q_{ab} = \hat\Theta^{-2}\hat{g}_{ab}$, i.e., $\hat{g}_{ab} = (\hat{\Theta}/\Theta)^{-2}g_{ab}$ and define $\theta = \hat{\Theta}/\Theta$. Inserting $\hat{\Theta} = \Theta \theta$ into the equation
\[
\hnabla_{\{a}\hnabla_{b\}} \hat{\Theta} = \hat\cR_{ab} \hat{\Theta}.
\]
and using the relation between $\nabla_a$ and $\hnabla_a$, when acting on a 1-form $v_a$ 
\[
  \nabla_av_b = \hnabla_av_b + 2\theta^{-1}\nabla_{\{a}\theta v_{b\}},
\]
yields after some calculation
\[
  \begin{multlined}
    \hat\cR_{ab} \hat{\Theta}
    = \Theta \hnabla_{\{a}\hnabla_{b\}} \theta + \theta \nabla_{\{a}\nabla_{b\}} \Theta.
\end{multlined}
\]
This implies the
\begin{prop}\label{prop:2}
  
Let $(\Theta,g_{ab}) \mapsto (\hat\Theta,\hat{g}_{ab}) = (\Theta\theta,\theta^2 g_{ab})$ be a conformal rescaling such that $q_{ab}=\Theta^{-2}g_{ab}$ remains unchanged. Then the distortion field $\hat\cR_{ab}$ is given by
  \begin{equation}
  \label{eq:9}
  \hat\cR_{ab} = \cR_{ab} + \theta^{-1}\hnabla_{\{a}\hnabla_{b\}} \theta = \cR_{ab} - \theta\nabla_{\{a}\nabla_{b\}} \theta^{-1}.
\end{equation}
\end{prop}

Finally, we consider a conformal factor $\Theta$ as before, such that $g_{ab}=\Theta^2q_{ab}$ is a metric with constant Gauß curvature. Let $\Phi$ be any solution of~\eqref{eq:8} and define $\phi = \Phi\Theta^{-1}$. Since both $\Theta$ and $\Phi$ satisfy~\eqref{eq:8} we find
\[
  \begin{aligned}
    0&=\nabla_{\{a}\nabla_{b\}} \Phi - \cR_{ab} \Phi = \phi \nabla_{\{a}\nabla_{b\}} \Theta + \Theta \nabla_{\{a}\nabla_{b\}} \phi + 2 \nabla_{\{a}\Theta \nabla_{b\}} \phi - \cR_{ab} \Theta \phi \\
    &=\Theta \left(\nabla_{\{a}\nabla_{b\}} \phi + 2 \Theta^{-1}\nabla_{\{a}\Theta \nabla_{b\}} \phi\right)
\end{aligned}
\]
Using the relation between $\nabla_a$ and $\del_a$, the covariant derivative of $q_{ab}$, acting on a 1-form $v_a$
\[
  \nabla_av_b = \del_av_b - 2\Theta^{-1}\del_{\{a}\Theta v_{b\}}
\]
we have shown the
\begin{prop}\label{prop:3}
  Let $\Theta$ be a conformal factor rescaling $g_{ab}$ to a round metric, and let $\Phi$ be an arbitrary solution of \eqref{eq:8}, then their quotient $\phi=\Phi\Theta^{-1}$ satisfies the equation
\begin{equation}
  \label{eq:10}
  \del_{\{a}\del_{b\}} \phi = 0.
\end{equation}
Conversely, every solution of~\eqref{eq:10} generates a solution $\Phi=\Theta\phi$ of~\eqref{eq:8}.
\end{prop}
In fact, if $\Phi$ happens to be another conformal factor rescaling to a round metric, then this result is also a consequence of the corollary~\ref{corr:1} applied in the case when $g_{ab}$ itself is already a round sphere metric being rescaled to another round sphere, since in that case the distortion field vanishes.

\section{The associated Minkowski space}
\label{sec:solving-phi-equation}

In this section, we focus on the solution space of~\eqref{eq:8}. In view of Proposition~\ref{prop:3}, this is equivalent to studying the solutions of~\eqref{eq:10}. Thus, we consider the situation described in the proposition and assume (without loss of generality) that the round metric~$q_{ab}=\Theta^{-2}g_{ab}$ has Gauß curvature equal to $1$.

We use the identity~\eqref{eq:1} which reduces to
\[
  \del_a\left(\Delta \phi - 2 \phi \right) = 2\del^b(\del_{\{a}\del_{b\}}\phi).
\]
Suppose, $\phi$ solves~\eqref{eq:10}. Then it follows that there is a constant $c$ such that
\[
  \Delta \phi - 2 \phi = 2c
\]
and, furthermore, that $\hat\phi = \phi-c$ is an eigenfunction of $\Delta$ with eigenvalue 2. Hence, every solution $\phi$ of~\eqref{eq:10} can be written in the form
\[
  \phi = \phi_2 + c
\]
where $\phi_2$ is any eigenfunction of $\Delta$ with eigenvalue 2, and $c$ is some constant.

Conversely, for any such function, the left-hand side of the identity vanishes, so that
\[
  \del^b(\del_{\{a}\del_{b\}}\phi) = 0.
\]
But since there are no divergence-free, symmetric and trace-free tensors on $S^2$, this equation implies~\eqref{eq:10}.

It is well known that the eigenspace of the Laplacian for the eigenvalue 2 is 3-dimensional. Thinking of $S^2$ as the unit-sphere in $\RR^3$, this eigenspace is spanned by the (restrictions of the) three Cartesian coordinate functions $(x,y,z)$. This means that the solution space of~\eqref{eq:10} is 4-dimensional, spanned by the functions
\[
  (X^A)_{A=0:3} := (1,x,y,z),
\]
the general solution of~\eqref{eq:10} is
\begin{equation}
  \phi = a_AX^A\label{eq:11}
\end{equation}
for some constants $a_A$, and the general solution of~\eqref{eq:8} can be written as $\Phi=\Theta\phi$.

Let us now consider the expression~\eqref{eq:7} for $\Phi$
\[
  \Sigma(\Phi,\Phi) = K \Phi^2 - \Phi\nabla^a\nabla_a\Phi + \nabla_a\Phi\nabla^a\Phi.
\]
Inserting $\Phi=\Theta \phi$ results in
\[
    \Sigma(\Phi,\Phi) 
    =  \phi^2\left(K\Theta^2  - \Theta\nabla^a\nabla_a\Theta + \nabla_a\Theta\nabla^a\Theta\right) - \Theta^2\phi\nabla^a\nabla_a\phi 
      + \Theta^2\nabla_a\phi\nabla^a\phi
\]
Rewriting this expression in terms of $\del_a$ and $q_{ab}$ and using the fact that the rescaled metric has unit Gauß curvature, we find that $\Sigma$ is conformally invariant 
\begin{equation}
  \Sigma(\Phi,\Phi) = \Sigma(\phi,\phi) = \phi^2 - \phi\del^a\del_a\phi + \del_a\phi\del^a\phi.\label{eq:12}
\end{equation}
With the explicit solutions for $\phi$, we can evaluate the expression explicitly. Keeping in mind, that $\Delta\phi = 2(\phi - a_0)$ we find
\[
  \Sigma(\phi,\phi) = \phi^2 - 2\phi(\phi - a_0) + \del_a\phi\del^a\phi.
\]
To compute the square of the gradient of $\phi$ we use $\del_a\phi = a_A\del_aX^A$ and the fact that the inverse metrics of $S^2$ and that of the embedding $\RR^3$ are related by 
\[
  q^{ab}\del_aX^A\del_bX^B = -\delta^{AB} + X^AX^B, \qquad \text{for } A,B=1,2,3.
\]
Thus, we find that
\[
  \begin{aligned}
    q^{ab} \del_a\phi\del_b\phi &= q^{ab}a_A\del_aX^Aa_B\del_bX^B = -\delta^{AB}\alpha_A\alpha_B + (X^Aa_A-a_0)^2 \\
    &= -(a_1^2 + a_2^2 + a_3^2) + (a_1x + a_2y + a_3z)^2 = -(a_1^2 + a_2^2 + a_3^2) + (\phi-a_0)^2.
\end{aligned}
\]
Finally, this results in the explicit expression 
\[
  \Sigma(\phi,\phi) = - \phi^2 + 2 a_0 \phi + \del_a\phi\del^a\phi  
  = a_0^2 - a_1^2 - a_2^2 - a_3^2.
\]
We summarise the previous discussion in the
\begin{theorem}\label{theorem:2}
  Let $(\cS,g_{ab})$ be a 2-dimensional spherical Riemannian manifold. Let $\Theta$ be a conformal factor to the unit-sphere.
  \begin{enumerate}[label=(\roman*)]
  \item The solution space of~\eqref{eq:8} is a four-dimensional real vector space~$\VV$. 
  \item Relative to $\Theta$, it is parametrized by the solutions $\phi$ of~\eqref{eq:10} given in~\eqref{eq:11}
  \item $\VV$ is naturally equipped with a Lorentzian metric~$\Sigma$, given by (the polarisation of) ~\eqref{eq:7},
  \item $\Phi$ is a conformal factor rescaling to a unit-metric iff $\Phi$ is a time-like unit-vector with respect to the inner product $\Sigma$.
\end{enumerate}
\end{theorem}

We mention some interesting points.
\begin{itemize}[label=--]
\item Not all elements of $\VV$ can serve as conformal factors. Those which produce constant Gauß curvature metrics are time-like vectors with respect to $\Sigma$, which implies that they are either strictly positive or strictly negative.
\item The representation of the solutions in terms of $\phi$ is relative to a specific conformal factor $\Theta$. This implies that the same solution $\Phi$ of~\eqref{eq:8} is represented in different ways relative to different conformal factors\footnote{However, it is a consequence of the conformal invariance of this situation that the ``coordinates'' $a_A$ assigned to $\Phi$ are independent of the conformal factor~$\Theta$.See the next section.}.
\item This fact is reminiscent of an affine space, whose elements (points or events) can be represented by position vectors once an origin has been selected. If we take this analogy more seriously, then we can interpret the (future) light-cone in $\VV$ as an ``affine'' manifold to which at every point --- a conformal factor $\Theta$ --- is associated a vector space of ``position vectors'' --- solutions $\phi$ of~\eqref{eq:10}--- which can be used to reach every other point. Clearly, the restriction to time-like vectors in $\VV$ removes the possibility to form arbitrary linear combinations.
\item This raises the question, as to what light-like or space-like solutions $\Theta$ of \eqref{eq:8} correspond to? Obviously, light-like $\Theta$'s correspond to rescalings between the sphere and the plane. At the same time, space-like $\Theta$'s relate the sphere to hyperbolic space, thus covering the three possibilities of the uniformization theorem of complex analysis~\cite{Weyl:1955,Jost:2006}. We have not explored this in more detail.
\end{itemize}

\section{Transformations}
\label{sec:invariance-group-1}

In this section, we discuss the transformations which leave the solution space of~\eqref{eq:8} invariant. These come in two flavours: since the equation is constructed using conformal properties of the sphere, all transformations which leave those invariant will also operate on these solutions. However, as we have seen, there are other possibilities of transforming solutions into solutions. We will discuss these first.

The previous discussions have shown that the essential properties are obtained already by considering only the unit-sphere. We have seen, that all solutions of~\eqref{eq:8} can be reduced to solutions of~\eqref{eq:10} where $\del_a$ is the connection compatible with a unit-metric $q_{ab}$. These are not unique. Therefore, we focus on the set of all functions $\phi$ for which a unit-metric exists with respect to which $\phi$ satisfies~\eqref{eq:10}.

Let $\hat{q}_{ab} = \theta^2q_{ab}$ be another unit-metric. Then the conformal factor $\theta$ satisfies~\eqref{eq:8}, i.e.,
\begin{equation}
  \label{eq:19}
  \hat{\del}_{\{a}\hat{\del}_{b\}} \theta = 0,
\end{equation}
with $\cR_{ab}=0$, since the curvature is constant. Now, suppose $\phi$ satisfies~\eqref{eq:10} with respect to $\del_a$. Then a short calculation shows that $\hat{\phi}=\theta \phi$ satisfies
\[
  \hat{\del}_{\{a}\hat{\del}_{b\}} \hat\phi = 0.
\]
Rewriting~\eqref{eq:19} in terms of $\del_a$ yields
\[
 0 = \del_{\{a}\del_{b\}} \theta - 2 \theta^{-1}\del_{\{a}\theta \del_{b\}}\theta = \theta^2\del_{\{a} \del_{b\}} \theta^{-1}. 
\]
It follows, that $\theta^{-1}$ satisfies~\eqref{eq:10} and can therefore be written as a linear combination of the first four spherical harmonics with respect to the metric $q_{ab}$. Thus, we have the following
\begin{prop}
  Let $q_{ab}$ be a unit-metric on $S^2$ and let $X^A=(1,x,y,z)$ be an orthonormal basis of the solution space $\VV$ of~\eqref{eq:10}. Then for every function $\hat\phi$ on $S^2$ of the form
  \[
    \frac{\phi_AX^A}{\theta_AX^A} \qquad \text{with } \eta^{AB}\theta_A\theta_B = 1
  \]
  there exists a unit-metric $\hat{q}_{ab}$ such that $\hat{\del}_{\{a}\hat{\del}_{b\}} \hat\phi = 0$. In particular, 
  \begin{enumerate}[label=(\roman*)]
  \item the metrics are related by $\hat{q}_{ab} = (\theta_AX^A)^{-2}q_{ab}$ ,
  \item the map $i:\VV \to \hat\VV$, $X^A \mapsto \hat{X}^A = (\theta_EX^E)^{-1}X^A$ is an isometry between the solution spaces of~\eqref{eq:10} with respect to the metrics $q_{ab}$ resp.\ $\hat{q}_{ab}$.
  \end{enumerate}
\end{prop}
We note, that this result reinforces the point of view that the solution spaces $\VV$ of \eqref{eq:10} with respect to any unit-metric are all isometric, the isometries being in 1-1 relationship with time-like future unit vectors. Relaxing this condition on $\theta$ enlarges the set of metrics from unit-metrics to general round metrics, i.e., metrics with arbitrary positive Gauß curvature. The solution spaces take the role of the spaces of position vectors attached to points in Minkowski space-time and the role of the translations between these points is played by the isometries $i$ between solution spaces.

A conformal Killing vector (field) is a vector field $Z^a$ on the sphere which satisfies the equation
\begin{equation}
  \label{eq:13}
  \nabla_{\{a} Z_{b\}} = 0.
\end{equation}
Since this equation is conformally invariant, we may consider it on the unit-sphere with respect to a unit-metric $q_{ab}$. We write it in the form
\begin{equation}
  \del_a Z_b = q_{ab} \beta + \eps_{ab} \omega,\label{eq:14}
\end{equation}
where $\beta$ and $\omega$ are some scalar functions on $S^2$. Taking another derivative and anti-symmetrizing results in the equation
\[
  -2\del_{[c}\del_{a]}Z_b = r_{cab}{}^eZ_e = -2 q_{b[a}\del_{c]}\beta + 2 \eps_{b[a}\del_{c]}\omega.
\]
Contracting and using $K=1$ yields
\[
  Z_a = \del_a\beta - \eps_a{}^e\del_e\omega,
\]
thus relating the Helmholtz decomposition of $Z_a$ to the degrees of freedom of the conformal Killing vector. Inserting back into~\eqref{eq:14} yields
\begin{equation}
  \del_a \del_b\beta - \eps_b{}^e\del_a\del_e\omega = q_{ab} \beta + \eps_{ab} \omega.\label{eq:15}
\end{equation}
The skew part and the trace of this equation give conditions on the two functions, namely
\[
  \Delta \beta = 2 \beta ,\qquad \Delta \omega = 2 \omega,
\]
i.e., they are both eigenfunctions of the unit-sphere Laplacian with eigenvalue 2. Since all these functions $f$ also satisfy the equation $\del_{\{a}\del_{b\}} f =0$, this implies that the trace-free part of~\eqref{eq:15} is identically satisfied. Now we can write the conformal Killing vector as
\[
  Z_a = \beta_A\del_aX^A - \eps_a{}^e\omega_A\del_eX^A.
\]
Here $\beta_A$ and $\omega_A$ are six arbitrary real constants\footnote{Note, that the constant functions, proportional to $X^0=1$, do not play a role here. Therefore, $A$ only takes values $A=1,2,3$.}. It is straightforward to represent $Z^a$ as a vector field $Z = Z^E\frac{\del}{\del X^E}$ in $\RR^3$ by computing its action on $X^E$. This results in
\begin{equation}
  Z^E =  -(\beta^E - \beta X^E) + \omega^AX^B \eps_{AB}{}^{E}.\label{eq:20}
\end{equation}
This shows that the boost part (given by the $\beta_A$) is a constant vector field modified so that it operates tangent to the sphere, while the rotation part (given by the $\omega_A$) is proportional to $X\del_Y - Y\del_X$ and similar, as expected. In this way, we obtain six vector fields generating a Lie algebra which is isomorphic to the Lorentzian Lie algebra.

A conformal Killing vector leaves the conformal class of a metric invariant, but not necessarily the metric itself, since
\[
  \Lie_Z q_{ab} = 2 \del_{(a}Z_{b)} = 2 \beta q_{ab}.
\]
This means that $Z^a$ also changes the connection and, hence, it also changes the equation~\eqref{eq:10} and its solution space $\VV$. The change in the connection due to $Z^a$ is obtained from the commutator acting on an arbitrary 1-form $v_b$
\[
  \left[\Lie_Z, \del_a \right] v_b = -2 \del_{\{a}\beta\, v_{b\}}.
\]
The Lie derivative of a solution $\phi=\phi_0+\phi_AX^A$ in the direction of $Z^a$ is easily found using the representation~\eqref{eq:20}
\[
  \begin{aligned}
    \Lie_Z\phi &= Z^E\frac{\del \phi}{\del X^E} =  -(\beta^E - \beta X^E)\phi_E + \omega^AX^B \eps_{AB}{}^{E}\phi_E\\
    &= \beta \phi  - \left\{\beta_E \phi^E \cdot 1 + \left(\phi_0\beta_B -  \eps_{EAB}\omega^A\phi^E \right)X^B\right\}.
\end{aligned}
\]
The infinitesimal change of $\phi$ due to a conformal Killing vector clearly comes in two parts. The first part is due to the change of the metric and it changes $\phi$ into a solution in an infinitesimally close solution space corresponding to an infinitesimally close metric. The second part is the action of an infinitesimal boost-rotation on (the components) of a 4-vector and transforms the solutions in $\VV$ among themselves.

As a final thought, we return to the group of Möbius transformations which leave the complex structure of the Riemann sphere invariant. The Möbius group is isomorphic to the identity component of the Lorentz group. It is a well-known fact that a Möbius transformation is fixed when the images of only three points of the sphere are designated. Thinking of these three points as the tips of three unit-vectors tells us that there exists a unique Möbius transformation which makes these unit-vectors point into the direction of a designated orthonormal basis. The corresponding Lorentz transformation is a combination of a rotation and a boost, the latter of which changing the metric by an appropriate rescaling.

In summary then, given a generic spherical 2-surface $(\cS,g_{ab})$ we can always find a conformal factor $\Theta$ which rescales $g_{ab}$ to a unit-metric. The function $\Theta$ is not fixed uniquely by the requirement that the Gauß curvature of the rescaled metric be equal to unity. The remaining freedom in the structure of the resulting unit-sphere corresponds exactly to the choice of a Minkowski tetrad.

\section{Canonical embedding into Minkowski space}
\label{sec:canonical-embeddings}

Finally, in this section, we want to mention an immediate consequence of the conformal properties of the sphere as formulated here. There are many situations in GR where foliations by spheres play an important role. Probably the most obvious case is the geometry of null-infinity, which can be regarded as a shear-free null hypersurface foliated by spherical ``cuts''. Some of the conformal properties discussed in the present paper have already been used in a somewhat different form in~\cite{Frauendiener:2022}. More general null hypersurfaces with spherical cross-sections should also be mentioned.

Foliations by spheres have been instrumental in the proof of the Riemannian version of the Penrose inequality~\cite{Penrose:1973}, which makes heavy use of a particular foliation by spheres suggested by Geroch~\cite{Geroch:1973} and discussed in more detail by Jang and Wald~\cite{Jang:1977}, and whose existence was proven by Huisken and Ilmanen~\cite{Huisken:2001} (see~\cite{Frauendiener:2001b} for a space-time version of this approach). Spherical foliations have also been used in numerical studies in GR: we only mention the most recent proposal for solving the constraints on space-like hypersurfaces using a hyperbolic system developed by \cite{Racz:2015} and the efforts to study global properties of space-times by solving the conformal field equations~\cite{Beyer:2017,Camden:2025,Frauendiener:2021a,Frauendiener:2025}.

However, we want to focus on the idea of a quasi-local energy-momentum. Here, one is interested in assigning a value of energy-momentum to a bounded region in space-time at a given instant of time. Usually, one thinks of the region as being contained inside a spherical boundary. The fundamental issue with this idea is this: to be physically meaningful, the energy-momentum should behave like a covariant Lorentzian 4-vector. However, this raises the immediate question: what is the corresponding vector space?

The extensive review article by Szabados~\cite{Szabados:2009} describes this problem and lists many proposed solutions. Those fall broadly into two categories: either one constructs finite-dimensional linear spaces by solving appropriate PDEs, or one seeks to connect to a reference structure that provides an appropriate Minkowski space. The prime example for the first approach is the original quasi-local mass construction by Penrose~\cite{Penrose:1982}, which used solutions of the twistor equation to construct energy-momentum and angular momentum (see also the ensuing variations by Tod and others~\cite{Tod:1983a,Tod:1986,Dougan:1991}). The first quasi-local mass definition using the second approach is due to Bartnik~\cite{Bartnik:1989}, who uses the asymptotic region of an asymptotically flat and stationary extension of the surface under consideration as the reference.

We found in the previous sections that every spherical space-like 2-surface comes equipped with its own Minkowski space, obtained as the space of solutions to a PDE intrinsically constructed from the sphere's conformal structure. These solutions can be used straightforwardly to construct an embedding of the surface in Minkowski space. The construction is as follows.

Let $\MM$ be Minkowski space-time with an origin $O$ selected. Let $(t^a,x^a,y^a,z^a)$ be an orthonormal frame and consider the future light-cone $\cC_O$ of $O$, i.e., the set of all events $P$ with future-pointing null position vectors $l^a= T t^a+ X x^a + Y y^a + Z z^a$ with Cartesian coordinates $(T,X,Y,Z)$ satisfying $T^2-X^2-Y^2-Z^2=0$. The intersection of $\cC_O$ with the space-like hyperplane $T=t_al^a=1$ yields a unit-sphere $S^+$ consisting of points with coordinates $(1,x,y,z)$.

Let $(\cS,g_{ab})$ be a generic sphere as discussed above and choose a conformal factor $\Theta$ rescaling $g_{ab}$ to a unit-sphere metric, and let $\Theta_{ab}$ be its associated distortion tensor. Choosing $\Theta$ as a unit time-like vector in $\VV$, we complement it to an orthonormal basis of solutions of~\eqref{eq:8}. These can be written as $\Theta \cdot(1,x,y,z) $, where $x$, $y$, and $z$ are appropriate eigenfunctions of the unit-sphere Laplacian as discussed in Sect.~\ref{sec:solving-phi-equation}. Now consider the map
\[
  i: \cS \to \MM, \qquad p \mapsto \Theta(p) \cdot(1,x(p),y(p),z(p)).
\]
This is an isometric embedding of $\cS$ into $\MM$. Its image can be described as being obtained by shifting every point on $S^+$, which has a position vector $l^a=(1,x,y,z)$, to a point with position vector $\Theta l^a$.

This construction depends on several choices: the conformal factor $\Theta$, the tetrads in both $\VV$ and $\MM$, as well as the choice of the unit-sphere $S^+\subset\MM$. It is easy to see that a change of tetrad in $\MM$, which is realized by the application of an orthochronous Lorentz transformation, corresponds to the application of a similar Lorentz transformation in $\VV$.

The effect of choosing a different $\Theta$ is related to the well-known geometrical picture of describing $S^+$ as the intersection of the future and past light cones of two different time-like separated vertices. Then the translation of one vertex with respect to the other results in the change of the induced metric on the intersection by a conformal factor of the form $\phi=a_AX^A$ as discussed in Sect.~\ref{sec:solving-phi-equation}.

The result of this discussion is not only that any spherical 2-surface can be embedded in Minkowski space, but also that it can be regarded as embedded in its intrinsically defined Minkowski space. It also shows that there is a canonical Minkowski space associated with every space-like 2-sphere in a general relativistic space-time. This fact is useful for discussions of quasi-local quantities, since this structure provides a natural background against which one can view and compare the many proposed constructions. Note that there are no dynamical space-time properties involved in the definition of the intrinsic Minkowski space, so that it is a purely kinematical property of the 2-surface.


\section{Summary and outlook}
\label{sec:conclusion}

The conformal geometry of the 2-sphere is well known. We have discussed it here from a somewhat unusual perspective, focusing on the set of conformal factors rescaling a given metric to a round metric. We found that all these conformal factors satisfy a linear PDE on the surface, expressed in terms of the co-curvature, a trace-free symmetric tensor that measures the sphere's deviation from roundness.

The set of such conformal factors does not by itself have a linear structure, but the solution space of the PDE does. It turned out that this space is a 4-dimensional real vector space which carries a natural Lorentzian inner product. This metric is defined by the transformation of the Gauß curvature under conformal rescalings. The conformal factors that rescale to a round metric foliate the (future) null cone into space-like hyperboloids according to the value of the constant Gauß curvature. The solution space can be parameterized with respect to a fixed conformal factor by the lowest four spherical harmonics, i.e., constants and eigenfunctions of the unit-sphere Laplacian with eigenvalue 2. 

Even though it is rather trivial, this property allowed us to interpret the structure as akin to an affine space, in which each point possesses its own local vector space of position vectors connecting it to any other point in the space.

The final observation we made was that each spherical 2-surface is equipped with its own canonically defined Minkowski space. In view of the many cases where spherical 2-surfaces occur in GR, in particular in foliations, it is natural to ask how the individual Minkowski spaces align across several spheres. This might have implications for the structure of null hypersurfaces, in particular of null-infinity. We have started the investigation in the latter case and will report on this in a forthcoming paper.

\printbibliography
\end{document}